\documentclass[12pt]{article}
\usepackage[doublespacing]{setspace}
\usepackage[pdftex]{graphicx}
\usepackage[pdftex]{color}
\usepackage{url}
\usepackage{epsfig,longtable}
\usepackage{makeidx}
\usepackage{varioref}
\usepackage{xr}
\usepackage{amsmath}
\usepackage{amsthm}
\usepackage{amssymb}
\usepackage{natbib}
\usepackage{setspace}
\usepackage{subfig}
\usepackage{multirow}
\usepackage{dsfont}
\usepackage{enumitem}
\usepackage{booktabs}
\usepackage{array}
\usepackage{caption}

\newtheorem{theorem}{Theorem}[section]
\newtheorem{proposition}{Proposition}[section]
\newtheorem{corollary}{Corollary}[section]

\numberwithin{equation}{section}

\begin{document}
\noindent \textsf{\Large Multivariate tests of association based on univariate tests}

\textsf{ Ruth Heller\footnote{\textit{Address for correspondence:} Department of
Statistics and Operations Research, Tel-Aviv university, Tel-Aviv,
Israel.  \ \textsf{E-mail:} ruheller@post.tau.ac.il.  } and Yair Heller}%
\\

\textsf{Abstract. \
For testing two random vectors for independence, we  consider testing whether the distance of one vector from a center point is independent from the distance of the other vector from a center point by a univariate test. 
In this paper we provide conditions under which it is enough to have a consistent univariate test of independence on the distances to guarantee that the power to detect dependence between the random vectors increases to one, as the sample size increases. These conditions turn out to be minimal. If the univariate test is  distribution-free, the multivariate test will also be distribution-free.    
If we consider multiple center points and aggregate the  center-specific univariate tests, the power may be further improved, and the resulting multivariate test may be distribution-free for specific aggregation methods (if the univariate test is distribution-free).  
We show that several multivariate tests recently proposed in the literature can be viewed as instances of this general approach. 
\medskip }

\noindent

\newpage

\section{Introduction}
Let ${ X}\in \Re^p$ and ${Y} \in \Re^q$ be random vectors, where $p$
and $q$ are positive integers.  The null hypothesis of independence 
is $H_0: F_{XY} =F_X F_Y,$
where the joint distribution of $({ X},{ Y})$ is denoted by
$F_{XY}$, and the distributions of ${ X}$ and ${ Y}$, respectively,
by $F_X$ and $F_Y$. If $X$ is a categorical variable with $K$ categories, then the null hypothesis of independence is  the null hypothesis in the $K$-sample problem,  $H_0: F_{1} =\ldots =  F_K,$ where $F_k, k\in \{1,\ldots,K\}$ is the distribution of $Y$ in category $k$.

The problem of testing for independence of random vectors, as well as the $K$-sample problem on a multivariate $Y$, against the  general alternative
 $H_1: F_{XY} \neq F_X F_Y,$  has received increased attention in recent years.  The most common approach is based on pairwise distances or similarity measures. See   \cite{Szekely07}, \cite{Gretton08}, \cite{Sejdinovic12}, and \cite{Heller12} for consistent tests of independence, and \cite{Hall02},  \cite{Szekely04}, \cite{Baringhaus04},  \cite{Rosenbaum05b}, \cite{Gretton06}, and \cite{Gretton12b} for  recent $K$-sample test. Earlier tests based on nearest neighbours include \cite{Schilling86}, and \cite{Henze88}.
 Another approach is to first reduce the multivariate data to a lower dimensional sub-space by  (random) projections, see  \cite{Cuesta-Albertos06} and \cite{Wei15}. 


 We suggest the following approach: first compute the  distances from a fixed center point, then apply any univariate test on the distances. If the univariate test is distribution-free, then so is the multivariate test.  In Section \ref{sec-mainresults} we show that a result  by  \cite{Rawat00} in integral geometry implies that  if  $H_0$ is false, then applying a univariate consistent test on distances from a single center point will result, for almost every center point,  in a multivariate test with power increasing to one as the sample size increases. 

 In Section \ref{sec-pooling} we show that considering the distances from $M>1$ points and aggregating the resulting statistics can also result in consistent tests, which may be more powerful than tests that consider a single center point.  Both distribution-free and permutation-based tests can be generated, depending on the choice of aggregation method and univariate test.  

In section \ref{subsec-existing} we draw the connection between these results and some known tests mentioned above. The tests of \cite{Hall02} and  of \cite{Heller12}  can be viewed as instances of this approach, where the  center point is a sample point, and all sample points are considered each in turn as a  center point, for a particular univariate test. 


\section{Main results}\label{sec-mainresults}
We  use the  following result  by \cite{Rawat00}. Let $B_d(x,r) = \{y\in \Re^d: \|x-y\| \leq r  \}$ be a ball centered at $x$ with radius $r$. A complex Radon measure $\mu$ on $\Re^d$ is  said to be of at most exponential-quadratic growth if there exist positive constants $A$ and $\alpha$ such that  $|\mu|(B_d(0, r)) \leq Ae^{\alpha r^2}$.
\begin{proposition}[\citet{Rawat00}] \label{thm-Rawat00}
Let $\Gamma\subset \Re^d$ be such that the only real analytic function (defined on an open set containing $\Gamma$) that vanishes on $\Gamma$, is the zero function.   Let $\mathcal{C} = \{B_d(x, r) : x \in \Gamma, r > 0\}$. Then for any complex Radon measure 
$\mu$ on $\Re^d$ of at most exponential-quadratic growth, if $\mu(C)=0$ for all $C\in \mathcal{C}$, then it necessarily follows that $\mu = 0$. 
\end{proposition}

For the two-sample problem, let  $Y\in \Re^q$ be a random variable with cumulative distribution $F_1$ in category $X=1$, and  $F_2$ in category $X=2$.  For $z \in \Re^q$, let $F'_{iz}$ be the cumulative distribution function of  $\|Y-z\|$ when $Y$ has cumulative distribution $F_i$, $i\in \{1,2\}$. We show that if the distribution of $Y$ differs across categories, then so does the distribution of the distance of $Y$ from almost every point $z$.  Therefore,
any univariate consistent two-sample test on the distances from $z$ results in a consistent  test of  the equality of the multivariate distributions $F_1$ and $F_2$, for almost every $z$. 
It is straightforward to generalize these results to $K>2$ categories.

\begin{theorem}\label{thm-2sample}
If $H_0: F_1=F_2$ is false,  then for every $z\in \Re^q$, apart from at most a set of Lebesgue measure 0,  there exists an $r>0$ such that $F'_{1z}(r) \neq F'_{2z}(r)$.
\end{theorem}
\begin{proof}
Suppose by contradiction, that there is a set $\Gamma\subseteq \Re^q$ with positive Lebesgue measure, such that for all $z\in \Gamma$, $F'_{1z}(r) = F'_{2z}(r)$ for all $r >0$. It follows that $\int_{y\in B_q(z,r)}dF_1(y) -\int_{y\in B_q(z,r)}dF_2(y) = 0$ for all $r>0$ and $z\in \Gamma$. 
Since $|F_1-F_2|\leq 1$, clearly $F_1-F_2$ is of at most exponential-quadratic growth. 
Moreover,  the only real analytic function that vanishes on $\Gamma$ is the zero function, since $\Gamma$ has positive Lebesgue measure. Therefore, it follows from Proposition \ref{thm-Rawat00} that  $F_1-F_2=0$, thus contradicting the fact that $H_0$ is false. 
\end{proof}
\begin{corollary}\label{cor-2sample}
For every $z\in  \Re^q$, apart from at most a set of Lebesgue measure 0, a consistent two-sample univariate test of the null hypothesis $H'_0: F'_{1z}=F'_{2z}$ will reject  $H_0: F_1=F_2$ with a power increasing to one as the sample size increases.
\end{corollary}    
\begin{proof}
If $H_0: F_1=F_2$ is false, then Theorem \ref{thm-2sample} guarantees that for every $z$, apart from at most a set of Lebesgue measure 0, 
the null univariate hypothesis, $H'_0: F'_{1z}=F'_{2z}$, is false. Since for such a $z$ the asymptotic power of a false null univariate hypothesis will be one for any consistent two-sample univariate test, the power of the multivariate test will be one. 
\end{proof}

For the multivariate test of independence, let  $X\in R^p$ and $Y\in \Re^q$ be two random vectors  with marginal distributions $F_X$ and $F_Y$, respectively, and with joint distribution $F_{XY}$. For $z = (z_x,z_y), z_x \in \Re^p, z_y \in \Re^q$, let $F'_{XYz}$ be the joint distribution of $(\|X-z_x \|, \|Y-z_y\|)$. Let  $F'_{Xz}$ and $F'_{Yz}$ be the marginal distribution of $\|X-z_x \|$ and $\|Y-z_y\|$, respectively.

\begin{theorem}\label{thm-indep}
If $H_0: F_{XY} = F_XF_Y$ is false, then for every $z_x\in \Re^{p}, z_y\in \Re^q$, apart from at most a set of Lebesgue measure 0, there exists $ r_x>0, r_y>0$, such that $F'_{XYz}(r_x,r_y) \neq F'_{Xz}(r_x)F'_{Yz}(r_y)$.
\end{theorem}
\begin{proof}
Suppose by contradiction, that there is a set $\Gamma \subseteq \Re^{p+q}$ with positive Lebesgue measure, such that for all $z\in \Gamma$, $F'_{XYz}(r_x,r_y) = F'_{Xz}(r_x)F'_{Yz}(r_y)$ for all $r_x>0, r_y>0$. It follows that for all $z\in \Gamma$ and $r_x>0, r_y>0$, 
$$
\int_{(\|x-z_x\|,\|y-z_y\|)\leq (r_x,r_y)}dF_{XY}(x,y) =  \int_{\|x-z_x\|\leq r_x}dF_{X}(x)\int_{\|y-z_y\|\leq r_y}dF_Y(y) .
$$
 It thus follows that for  all $z\in \Gamma$ and any $r>0$,
 \begin{equation}\label{eq-1-thm-indep}
\int_{(\|(x,y)-(z_x,z_y)\|)\leq r}dF_{XY}(x,y) =  \int_{\|(x,y)-(z_x,z_y)\|\leq r}dF_{X}(x)dF_{Y}(y).
\end{equation}
However, from  Theorem \ref{cor-2sample}, with $F_1 = F_{XY}$ and $F_2 = F_XF_Y$, it follows that for all $z\in \Gamma$, apart from a set of Lebesgue measure 0, there exists an $r>0$ such that  $$\int_{(\|(x,y)-(z_x,z_y)\|)\leq r}dF_{XY}(x,y) \neq  \int_{\|(x,y)-(z_x,z_y)\|\leq r}dF_{X}(x)dF_{Y}(y),$$ 
thus contradicting (\ref{eq-1-thm-indep}). 
\end{proof}

\begin{corollary}\label{cor-thm-indep}
For every $z \in  \Re^{p+q}$, apart from at most a set of Lebesgue measure 0, a consistent  univariate test of inependence of the null hypothesis $H'_0: F'_{XYz}  = F'_{Xz}F'_{Yz}$ will reject  $H_0: F_{XY}=F_{X}F_{Y}$ with a power increasing to one as the sample size increases.
\end{corollary}    
\begin{proof}
If $H_0: F_{XY}=F_{X}F_{Y}$ is false, then Theorem \ref{thm-indep} guarantees that for  every $z$, apart from at most a set of Lebesgue measure 0,
the null univariate hypothesis, $H'_0: F'_{XYz}  = F'_{Xz}F'_{Yz}$, is false. Since for such a $z$ the asymptotic power of a false null univariate hypothesis will be one for any consistent  univariate test of independence, the power of the multivariate test will be one. 
\end{proof}

We have $N$ independent copies $({ x}_i,{ y}_i)$ ($i=1,\ldots,N$) from the joint distribution $F_{XY}$. The above results motivate the following two-step procedure for the  multivariate test. 
For the $K$-sample test,  $x_i\in \{1,\ldots, K\}$ determines the category and $y_i\in  \Re^q$ is the observation in category $x_i$, so the two-step procedure is to first choose $z\in \Re^q$ and  then to  apply a univariate $K$-sample consistent test on $(x_1,\|y_1-z \|), \ldots, (x_N,\|y_N-z \|)$.  Examples of such univariate tests include the classic Kolmogorov-Smirnov and Cramer-von Mises tests. 
For the test of independence,  the two-step procedure is to first choose $z_x\in \Re^p$ and $z_y\in \Re^q$, and then to apply a univariate consistent test of independence on $(\|x_1-z_x\|,\|y_1-z_y \|), \ldots, (\|x_N-z_x\|,\|y_N-z_y \|)$.  An example of such a univariate test is the classic test of \cite{Hoeffding48}.  See \cite{Hellera} for novel  $K$-sample and independence tests and a review of existing distribution-free univariate tests. Note that  the consistency of a univariate test may be satisfied only under some assumptions  on the distribution of the distances. For example, the consistency of \cite{Hoeffding48} follows if the densities of $\|X-z_x\|$ and $\|Y-z_y\|$ are continuous. 

A great advantage of the two-step procedure is the fact that it has the same computational complexity as the univariate test. 
For example, if one chooses to use Hoeffding's univariate independence permutation test \citep{Hoeffding48}  , then the total complexity is only $O(N\log N)$, which is the cost of computing the test statistic. The $p$-value can be extracted from a look-up table since Hoeffding's test is distribution-free. For comparison, note that the  computational complexity of the multivariate permutation tests of \cite{Szekely07} and  \cite{Heller12} is $O(BN^2)$, and  $O(BN^2\log N)$, respectively, where $B$ is the number of  permutations. For many univariate tests the asymptotic null distribution is known, thus it can be used to compute the significance efficiently without resorting to permutations, which are typically required for assessing the multivariate significance.  

Another advantage of the two-step procedure is the fact that the test statistic may be estimating an easily interpretable population value. The univariate test statistics often converge to easily interpretable population values, which are often between 0 and 1. These values carry over to provide meaning to the new multivariate statistics.

In practice, the choice of the center value from which the distances are measured can be important. In the next Section \ref{sec-pooling}, we suggest generalizations of the above two-step approach that use multiple center points.

\section{Pooling univariate tests together}\label{sec-pooling}

We need not rely on a single $z\in \Re^{p+q}$ (or a single $z\in \Re^q$ for the $K$-sample problem). If we apply a consistent univariate test using many points  $z_i$ for $i=1,\ldots,M$ as our center points, where the test is applied on the distances of the  $N$ sample points from the center point, we obtain $M$ test-statistics and corresponding $p$-values,  $p_1,\ldots,p_M$. 

We can use the $p$-values or the test statistics of the univariate tests to design consistent multivariate tests. We suggest three useful approaches. The first approach is to combine the $p$-values, using a  combining function $f: [0,1]^M\rightarrow [0,1]$. Common combining functions include $f(p_1,\ldots,p_M) = \min_{i=1,\ldots,M} p_i$, and $f(p_1,\ldots,p_M) =-2\sum_{i=1}^M \log p_i$. 
The test statistic  $f(p_1,\ldots,p_M) = \max_{i=1,\ldots,M} p_i$ may also have excellent power in applications where the univariate test on the distances from almost all points has power.  The second approach is to combine the univariate test statistics, by taking the  average, maximum, or minimum statistic. These aggregation methods can result in test statistics which converge to meaningful population values, see equation (\ref{eq-KSmax}) below for multivariate tests based on the  univariate Kolmogorov-Smirnov two sample test  \citep{Kolmogorov41}.  We note that if the univariate tests are distribution-free then taking the maximum (minimum) $p$-value is equivalent to taking the minimum (maximum) test statistic (when the test rejects for large values of the test statistic).  The significance of the combined $p$-value or the combined test statistic can be computed by a permutation test.

A drawback of the two approaches above is that the distribution-free property of the univariate test does not carry over to the multivariate test. In our third approach, we consider the set of $M$ $p$-values as coming from the family of $M$ null hypotheses, and then apply a valid test of the global null hypothesis that all $M$ null hypotheses are true. Let $p_{(1)}\leq \ldots \leq p_{(M)}$ be the sorted $p$-values. The simplest valid test for any type of dependence is the Bonferroni test, which will reject the global null if $M p_{(1)} \leq \alpha$. Another valid test is the test of  \cite{Hommel1983}, which rejects if $\min_{j\geq 1} \{M (\sum_{l=1}^M1/l) p_{(j)}/j\}\leq \alpha$. This test statistic was suggested independently in a multiple testing procedure for false discovery rate control under general dependence in \cite{Benjamini01}.     
This approach is computationally much more efficient than the first two approaches, since no permutation test is required after the computation of the univariate $p$-values, but it may be less powerful. Clearly, if the univariate test is distribution free, the resulting multivariate test  is also distribution free.

As an example we shall prove that when using the Kolmogorov-Smirnov two sample test  as the univariate test, all of the pooling methods above result in consistent multivariate two-sample tests.  Let $KS(z) = \sup_{d\in \Re}|F'_{1z} (d) - F'_{2z}(d)|$ be the population value of the univariate Kolmogorov-Smirnov two sample test statistic comparing the distribution of the distances. Let $N$ be the total number of independent observations, where we assume for simplicity an equal number of observations from $F_1$ and from $F_2$. 
\begin{theorem}\label{thm-poolingKS}
Let ${z_1,\ldots, z_M}$ be a sample of center points from an absolutely continuous distribution with probability measure $\nu$, whose support $S$ has a positive Lebesgue measure in $\Re^q$. Let $KS_N(z_i)$ be the empirical value of $KS(z_i)$ with corresponding $p$-value $p_i$, $i=1,\ldots,M$. Let $p_{(1)}\leq \ldots \leq p_{(M)}$ be the sorted $p$-values.  Assume that the distribution functions $F_1$ and $F_2$ are continuous.  For $M = o(e^N)$, if $H_0: F_1=F_2$ is false,  then $\nu$-almost surely, as  $N\rightarrow \infty$, the power will increase to one for the following level $\alpha$ tests: 
\begin{enumerate}
\item the  permutation test using the  test statistics $S1=\max_{i=1,\ldots,M}\{KS_N(z_i)\}$ or  $S2=p_{(1)}$.
\item the test based on Bonferroni, which rejects $H_0$ if $M p_{(1)}\leq \alpha$.
\item for $M\log M = o(e^N)$, the test based on  Hommel's global null  $p$-value, which rejects $H_0$ if $\min_{j=1,\ldots,M}\left\{M (\sum_{l=1}^M1/l) p_{(j)}/j\right\}\leq \alpha$.
\item the  permutation tests using the statistics $T1=\sum_{i=1}^M KS_N(z_i)$ or  $T2=-2\sum_{i=1}^M\log p_i$.
\end{enumerate}

\end{theorem}

\begin{proof}
 Proving item 2 will prove item 1 since if the Bonferonni adjusted $p$-value is consistent then so is  the permutation test based on the minimum p-value (or maximum statistic), which has necessarily a smaller $p$-value than $Mp_{(1)}$. 
We need to show that the probability of rejection goes to one when $H_0$ is false. 
According to Corollary \ref{cor-2sample} when $H_0$ is false, $\nu$-almost surely any point $z_i$ offers a consistent univariate test. Therefore, $\nu$-almost surely,
$\sup_{i=1,\ldots,M} KS(z_i) \geq KS(z_1)>0$. 
Let $d_0$ be a distance such that $|F'_{1z_1}(d_0)-F'_{2z_1}(d_0)|=c>0$. 

Let $F'_{iz_1N}$ be the empirical cumulative distribution function based on $N/2$ sampled distances from $F'_{iz_1}$, $i\in \{ 1,2\}$. The test statistic is bounded away from zero:
\begin{eqnarray}
&&  {\rm pr}\{\sup_{i=1, \ldots, M} {KS_N(z_i)}>c/2)\}\geq  {\rm pr}\{KS_N(z_1)>c/2\} =  {\rm pr}\{\sup_d|F'_{1z_1N}(d)-F'_{2z_1N}(d)| > c/2\} \nonumber \\
&&  \geq  {\rm pr}\{|F'_{1z_1N}(d_0)-F'_{2z_1N}(d_0)| > c/2\} \nonumber \\
&& \geq 
 {\rm pr}\{|F'_{1z_1N}(d_0)-F'_{1z_1}(d_0)|<c/4\} {\rm pr}\{|F'_{2z_1N}(d_0)-F'_{2z_1}(d_0)|<c/4\}  \geq (1-2e^{-Nc^2/8})^2 \nonumber 
\end{eqnarray}
where in the last row, the first inequality follows since if $|F'_{1z_1N}(d_0)-F'_{1z_1}(d_0)|<c/4$ and 
$|F'_{2z_1N}(d_0)-F'_{2z_1}(d_0)|<c/4$, given that $|F'_{1z_1}(d_0)-F'_{2z_1}(d_0)|=c$, it implies that 
$|F'_{1z_1N}(d_0)-F'_{2z_1N}(d_0)| > c/2$, and the last inequality is the Dvoretzky–-Kiefer–-Wolfowitz inequality \citep{DKW56}. Therefore, when $H_0$ is false, the probability that the statistic is greater than $c/2$ goes to 1 as $N\rightarrow \infty$.

When $H_0$ is true, let $F_{z}'$ denote the common distribution function of $\|Y-z \|$.  For each $z \in \{ z_1,\ldots,z_M\}$, 
\begin{eqnarray}
&&  {\rm pr}\{KS_N(z)>c/2\} =  {\rm pr}\{\sup_d|F'_{1zN}(d) - F'_{z}(d) + F'_{z}(d) - F'_{2zN}(d)| > c/2\} \nonumber \\
&&\leq  {\rm pr}\{\sup_d|F'_{1zN}(d) - F'_{z}(d)| + \sup_d|F'_{2zN}(d) - F'_{z}(d)| > c/2\} \nonumber \\ && \leq  {\rm pr}\{\sup_d|F'_{1zN}(d) - F'_{z}(d)| >c/4\} +  {\rm pr}\{\sup_d|F'_{2zN}(d) - F'_{z}(d)| >c/4\} \leq 4e^{-Nc^2/8}, \label{eq-dkw-forpv}
\end{eqnarray}
where the last inequality follows from the Dvoretzky–-Kiefer–-Wolfowitz inequality.
It follows from (\ref{eq-dkw-forpv}) that the Bonferonni adjusted p-value is  bounded above  by $4Me^{-Nc^2/8}$, and therefore goes to zero  as $N\rightarrow \infty$ for $M=o(e^N)$, proving consistency.

For item 3,  the proof is very similar. Hommel's global null $p$-value is at most $M (\sum_{l=1}^M1/l) p_{(1)}$,  and as in the proof for item 2 it is bounded above by  $4M(\sum_{l=1}^M1/l)e^{-Nc^2/8}$, which goes to zero as $N\rightarrow \infty$ for $M\log M = o(e^N)$.

For item 4, let $z_0 \in \Re^q$ be a center point sampled from $\nu$. 
When $H_0$ is false, $\nu$-almost surely $KS(z_0)=c>0$. By Lebesgue's density theorem $\nu$-almost surely there exists an $\epsilon$ such that if $r<\epsilon$ then at least half of the ball $B_q(z_0,r)$ is within the support $S$. 
Since $F_1$ and $F_2$ are continuous, $KS(z)$ is a continuous function of $z$. Therefore,  there exists an $\epsilon'<\epsilon$ such that  $KS(z)>c/2$  for all  $z\in B_q(z_0,\epsilon')\cap S$. 
Similar arguments to those for item 2 show that $\nu$-almost surely for any $z_i\in S\cap B_q(z_0,\epsilon')$,
\begin{equation} \label{eq-thm3item4}   
Pr(KS_N(z_i)<c/4)<4e^{-Nc^2/32}. 
\end{equation} 
Therefore, $Pr(\cup_{z_i\in S\cap B_q(z_0,\epsilon')}  KS_N(z_i)<c/4)<4Me^{-Nc^2/32}.$
Since $\nu$-almost surely ${\rm pr}\{Z \in   S\cap B_q(z_0,\epsilon')\}>0$, then $\nu$-almost surely with probability going to one $T1$ is $O(M)$, as long as $M=o(e^N)$.  On the other hand when $H_0$ is true,  $E(KS_N(z))=O(1/\sqrt{N})$, see for example  \citet{Marsaglia2003}.  Therefore, $E(T1)=O(M/\sqrt{N})$, and by Markov's inequality the permutation test based on $T1$ will have  $\nu$-almost surely  power increasing to one as the sample size increases.
For the test based on $T_2$, from  equations (\ref{eq-thm3item4}) and (\ref{eq-dkw-forpv}) it follows that for $N$ large enough   $p_i<4e^{-Nc^2/32}$ for $z_i \in S\cap B_q(z_0,\epsilon')$, $i=1,\ldots,M$. Therefore, for $N$ large enough $-2\sum_{i=1}^M \log p_i$ is  greater than $O(NM){\rm pr}\{Z \in   S\cap B_q(z_0,\epsilon')\}$. On the other hand, when $H_0$ is true $P_i$ is uniformly distributed, so  $E(-2\sum_{i=1}^M \log P_i)$ is $O(M)$. By  Markov's inequality the permutation test based on $-2\sum_{i=1}^M \log p_i$ will have $\nu$-almost surely power increasing to one as the sample size increases.
\end{proof}

The test statistics $S1$ and $T1/M$ converge to  meaningful population quantities, 
\begin{eqnarray}\label{eq-KSmax}
\lim_{N, M \rightarrow \infty} S_1 = \lim_{M\rightarrow \infty} \max_{z_1,\ldots,z_M} KS(z) = \sup_{z\in S}KS(z), \nonumber \\
\lim_{N, M \rightarrow \infty} T_1/M = \lim_{M\rightarrow \infty} \sum_{i=1}^M KS(z_i)/M = E\{KS(Z) \}, 
\end{eqnarray}
where the expectation is over the distribution of the center point $Z$.

Arguably, the most natural choice of center points is the sample points themselves. 
Interestingly,  if the univariate test is a $U$-statistic \citep{Hoeffding48b} of order $m$, then the resulting multivariate test is a $U$-statistic of order $m+1$, if each sample point acts as a center point, and the univariate test statistics are averaged. The proof is as follows. Denote $B(a,b)$ for the binomial coefficient $a$ choose $b$. If the univariate test statistic $T_{N-1}$ is a $U$-statistic, then it can be written as 
$
T_{N-1}= \sum_{C_{N-1,m}}h\{ (u_{j_1},v_{j_1}),\ldots,(u_{j_m},v_{j_m}) \}/B(N-1, m),$ where $h$ is a symmetric function, $(u_{j_1},v_{j_1}),\ldots,(u_{j_m},v_{j_m})$ is a subset of size $m$  from  a sample of size $N-1$,   and $C_{N-1,m}$ is the set of all such subsets of size $m$.   
The multivariate test statistic is therefore 
$$
\sum_{C_{N,m+1}} f\{(x_{j_1},y_{j_1}), \ldots, (x_{j_{m+1}}, y_{j_{m+1}}) \}/B(N, m+1), 
$$
where $f\{(x_1,y_1), \ldots, (x_{m+1}, y_{m+1}) \}$ is the symmetric function 
$$
\frac 1{m+1}[ h\{(\|x_{k}-x_{1} \| ,\| y_{k}-y_1 \|), k=2,\ldots, m+1 \} + \ldots +  h\{(\|x_{k}-x_{m+1} \| ,\| y_{k}-y_{m+1} \|), k=1,\ldots, m\}]. 
$$


\section{Connection to existing methods} \label{subsec-existing}
 We are aware of two multivariate test statistics of the above-mentioned form: aggregation of the univariate test statistics on the distances from center points. The tests are the two sample test of \cite{Hall02} and the independence test of \cite{Heller12}. Both these tests use the second pooling method mentioned above by summing up the univariate test statistics. Furthermore, both these tests use the $N$ sample points as the center points (or $z$'s) and perform a univariate test on the remaining $N-1$ points. Indeed, \cite{Hall02} recognized that their test can be viewed as summing up univariate Cramer von-Mises tests on the distances from each sample point. We shall show that the test statistic of \cite{Heller12} can  be viewed as aggregation by summation of the univariate weighted Hoeffding independence test suggested in \cite{Thas04}. 
 
 \cite{Heller12} presented a permutation test based on the test statistic $\sum_{i=1}^N \sum_{j=1, j\neq i}^N S(i,j)$, where $S(i,j)$ is the Pearson test score for the $2\times 2$ contingency table for the random variables $I(\|X-x_i\|\leq \|x_j-x_i\|)$ and $I(\|Y-y_i\|\leq \|y_j-y_i\|)$, where $I(\cdot)$ is the indicator function. Since $\|X-x_i\|$ and $\|Y-y_i\|$ are univariate random variables, $S(i,j)$ can also be viewed as the test statistic for the test of independence between $\|X-x_i\|$ and $\|Y-y_i\|$, based on the $2\times 2$ contingency table induced by the $2\times 2$ partition of $\Re^2$ about the point $(\|x_j-x_i\|,\|y_j-y_i\| )$ using the $N-2$ sample points $(\|x_k-x_i\|,\|y_k-y_i\| ), k=1,\ldots, N, k\neq i, k\neq j$. \cite{Thas04} showed that the  statistic that   sums the Pearson test statistics over all  $2\times 2$ partitions of $\Re^2$ based on the observations, results in a consistent test of independence for univariate random variables.  The test statistic of \cite{Thas04} on the sample points $(\|x_k-x_i\|,\|y_k-y_i\| ), k=1,\ldots, N, k\neq i$, is therefore $\sum_{j=1, j\neq i} S(i,j)$. The multivariate test statistic of \cite{Heller12} aggregates by summation the univariate test statistics of  \cite{Thas04}, where the $i$th univariate test statistic is based on the $N-1$ distances of $x_k$ from $x_i$, and the $N-1$ distances of $y_k$  from $y_i$, for $k=1,\ldots, N, k\neq i$.

Of course, not all known consistent multivariate tests belong to the framework defined above. As an interesting example we discuss the energy test of  \cite{Szekely04} and \cite{Baringhaus04} for the two-sample problem. Without loss of  generality, let $y_1,\ldots, y_{N_1}$ be the observations from $F_1$, and $y_{N1+1},\ldots, y_{N}$ be the observations from $F_2$.
The test statistic is  
\begin{eqnarray}
  \mathcal{E} =\frac{N_1N_2}{N_1+N_2}\left(\frac{2}{N_1N_2}\sum_{l=1}^{N_1}\sum_{m=N_1+1}^{N}
\| y_{l}-y_{m} \| -\frac{1}{N_1^2}\sum_{l=1}^{N_1}\sum_{m=1}^{N_1} \| y_{l}-y_{m} \|-\frac{1}{N_2^2}\sum_{l=N_1+1}^{N}\sum_{m=N_1+1}^{N}\| y_{l}-y_{m} \| \right), \nonumber
\end{eqnarray}
where $\| \cdot \|$ is the Euclidean norm. It is easy to see that $ \mathcal{E} =\sum_{i=1}^N S_i$, where the univariate score is $S_i = \left\{\frac{1}{N_1}\sum_{m=1}^{N_1}\| y_{i}-y_{m} \|-\frac{1}{N_2}\sum_{m=N_1+1}^{N_2}\| y_{i}-y_{m}  \|\right\}w(i)$, $w(i) = -\frac {N_2}{N}$ if $i\leq N_1$ and $w(i) = \frac {N_1}{N}$ if $i> N_1$,  for  $i\in \{ 1,\ldots, N\}$.
The statistic $S_i$ is not an omnibus consistent test statistic, since a test based on $S_i$ will have no power to detect difference in distributions with the same expected  distance from $y_i$ across groups. However,  the energy test is omnibus consistent. 





\section*{Acknowledgement}
We thank Elchanan Mossel for useful discussions of the main results.


\end{document}